\documentclass[sn-mathphys,Numbered]{sn-jnl}

\usepackage{graphicx}%
\usepackage{multirow}%
\usepackage{amsmath,amssymb,amsfonts}%
\usepackage{amsthm}%
\usepackage{mathrsfs}%
\usepackage[title]{appendix}%
\usepackage{xcolor}%
\usepackage{textcomp}%
\usepackage{manyfoot}%
\usepackage{booktabs}%
\usepackage{algorithm}%
\usepackage{algorithmicx}%
\usepackage{algpseudocode}%
\usepackage{listings}%
\usepackage{amssymb,amsmath,latexsym}
\usepackage{mathrsfs}
\usepackage[shortlabels]{enumitem}
\usepackage{color}
\usepackage{comment}
\usepackage{xcolor}
\usepackage{stackengine}
\usepackage{calc}
\usepackage{physics}
\usepackage{accents}
\usepackage[most]{tcolorbox}
\allowdisplaybreaks

\newcommand{\wt}{\widetilde}

\newcommand{\mb}{\mathbb}
\newcommand{\mc}{\mathcal}

\theoremstyle{thmstyleone}%
\newtheorem{theorem}{Theorem}
\newtheorem{definition}{Definition}
\newtheorem{lemma}{Lemma}

\theoremstyle{thmstyletwo}%
\newtheorem{remark}{Remark}%
\newtheorem{assumption}{Assumption}

\begin{document}
\title[A note on the stabilizer formalism via noncommutative graphs]{A note on the stabilizer formalism via noncommutative graphs}


\author*[1,2]{\fnm{Roy} \sur{Araiza}}\email{raraiza@illinois.edu}

\author[1]{\fnm{Jihong} \sur{Cai}}
\author[1]{\fnm{Yushan} \sur{Chen}}
\author[1]{\fnm{Abraham} \sur{Holtermann}}
\author[1]{\fnm{Chieh} \sur{Hsu}}
\author[1]{\fnm{Tushar} \sur{Mohan}}
\author*[1,3]{\fnm{Peixue} \sur{Wu}}\email{p33wu@uwaterloo.ca}
\author[1]{\fnm{Zeyuan} \sur{Yu}}

\affil*[1]{\orgdiv{Department of Mathematics}, \orgname{University of Illinois at Urbana-Champaign}, \country{USA}}

\affil[2]{\orgdiv{IQUIST}, \orgname{University of Illinois at Urbana-Champaign}, \country{USA}}

\affil[3]{\orgdiv{Institute for Quantum Computing}, \orgname{University of Waterloo}, \country{Canada}}


\abstract{In this short note we formulate a stablizer formalism in the language of noncommutative graphs. The classes of noncommutative graphs we consider are obtained via unitary representations of finite groups, and suitably chosen operators on finite-dimensional Hilbert spaces. Furthermore, in this framework, we generalize previous results in this area for determining when such noncommutative graphs have anticliques. }

\maketitle

\section{Introduction}
Since they were first abstractly characterized by Choi and Effros in \cite{Choi}, operator systems have had profound applications on many areas of mathematics, and more recently, on quantum information theory, see e.g. \cite{duan2012,weaver2017, Brannan2022}. Such a framework led to generalizations of Lovász' famous theta function $\vartheta$ \cite{duan2012}, nonlocal games \cite{Brannan2020}, the study of correlation sets \cite{Araiza1, Araiza2}, and extensions of classical graph invariants  \cite{Paulsen13}. Our work will focus on a special class of finite-dimensional operator systems called \emph{noncommutative graphs}, which were first introducted in the seminal work of Duan-Severini-Winter in \cite{duan2009}. Using this special class of operator systems, Duan,Severini and Winter were able to extend the notions of classical, quantum, and entanglement assisted capacities. Noncommutative graphs are obtained by taking a Kraus representation of a quantum channel $\Phi: \mathbb B(H_1) \to \mathbb B(H_2), \Phi \rho = \sum E_i \rho E_i^\dag$, and considering the subspace $\mathcal V:= span\{E_i^\dag E_j: 1 \leq i,j \leq n\}.$ Two important properties of $\mathcal V$ are what characterizes an operator system. First note $\mathcal V = \mathcal V^\dag,$ which is to say $\mathcal V$ is \emph{self-adjoint}, and second is that $I \in \mathcal V$, which is to say that $\mathcal V$ is unital. Due to the aforementioned paper of Choi-Effros, every operator system may be concretely represented as a self-adjoint unital subspace of $\mathbb B(H)$. 
Though not considered in this framework, noncommutative graphs are the essential ingredient in the famous Knill-Laflamme subspace condition from quantum error correction (cite Knill-Laflamme). In particular, given a quantum channel $\Phi$ with Kraus representation $\{E_i\}_{i=1}^r, E_i \subset \mathbb B(H_1: H_2)$, then $\Phi$ is correctable if and only if \begin{align}
    PE_i^\dag E_jP = \lambda_{ij} P, \label{eq: Knill Laflamme}
\end{align} for every error $E_i \in \mathbb B(H_1: H_2)$, where $P: \mathcal H \to \mathcal C$ is the projection onto the codespace $\mathcal C$,  and $\lambda = [\lambda_{ij}] \in \mathbb M_r$ is a hermitian matrix.Equivalently, this may be expressed as \begin{align}
    P\mathcal V P = \mathbb C P. 
\end{align} Given a noncommutative graph $\mathcal V$, if Equation~\eqref{eq: Knill Laflamme} is satisfied, we will call $P$ an \emph{anticlique} for $\mathcal V$. \indent With this framework in mind, a natural question is the following: \emph{for which classes of noncommutative graphs do there exists anticliques?} This question led to a series of papers (\cite{amosov2017, amosov2018_1, amosov2018_2}) in which the authors answer this question for particular examples of classes of noncommutative graphs. In particular, in \cite{amosov2018_2}, the authors prove that $span \{U_\varphi M_o U_{\varphi}: \varphi \in \mathbb T\}$ is a noncommutative graph with anticliques $\{P_s: 1\leq s \leq d\}$, where $U_\varphi = \sum_{s=1}^d e^{i\varphi s} P_s$ is a unitary representation of $\mathbb T$. In \cite{weaver2017}, using combinatorial techniques, it was proven that for $\text{dim}\ \mc H = d$ and $k\le d$, if the noncommutative graph $\mc V$ satisfies 
 \begin{align}\label{bound: general}
 \text{dim}\ \mc V(\text{dim}\ \mc V + 1)\le \frac{d}{k},
 \end{align}
then there exists $\mc C \subseteq \mc H$ with $\text{dim}\ \mc C=k$ such that $\text{dim}\ P_{\mc C}\mc V P_{\mc C} = 1$. Moreover, if $\mc V$ is given by 
\begin{align*}
\mc V:= \text{span}\{E_1,\cdots, E_m\},\ [E_i,E_j]=0,\ \forall i,j\le m
\end{align*}
and
\begin{equation}\label{bound: commuting}
\dim\ \mc V \le \frac{d-k}{k-1},
\end{equation}
then there exists $\mc C \subseteq \mc H$ with $\text{dim}\ \mc C=k$ such that $\text{dim}\ P_{\mc C}\mc V P_{\mc C} = 1$. In  \cite{amosov2017} it was pointed out that the above upper bound is not sharp when the underlying Hilbert space $\mc H$ has a tensor product structure. In fact, if $\mc H = \mc H_1\otimes \mc H_2$ with $\text{dim}\ \mc H_1 = \text{dim}\ \mc H_2 = n>2$, a concrete noncommutative graph $\mc V$ with $\text{dim}\ \mc V = 2n(n-1)+1$ and a code space with dimension $n$ were given, which violates the above bounds \eqref{bound: general} and \eqref{bound: commuting}. We lend our hand in answering this question by weakening the assumptions in \cite{amosov2018_2} and therefore constructing a larger class of noncommutative graphs which exhibit anticliques. To this end, given a finite group $G$, let $\pi: G \to \mathbb B(H)$ be a unitary representation, $M_o \in \mathbb B(H)$, and consider the subspace $\mathcal V_{M_o}:= span \{\pi(g) M_o \pi(g): g \in G\}$. For any $g\in G$, let $P_{i_{g}}$ be the projector onto the $i$-th eigenspace of $\pi(g)$. Then our first main result is as follows: 
\begin{theorem}
    Suppose $\mathcal V_{M_o}$ is a noncommutative graph and $G$ is Abelian. Then for any sequence $\{i_g \in J_g\}_{g\in G}$, if $$P = \prod_{g\in G}P_{i_g}$$
has rank no less than 2, then $P$ is an anticlique for $\mc V$.
\end{theorem}

Our next and final theorem is obtaining a stabilizer formalism in the language of noncommutative graphs. In what follows, let $\mathcal P_n$ denote the Pauli group acting on $2^n$ qubits.

\begin{theorem}\label{Stabilizer: new}
For any Abelian subgroup $G \subseteq \mc P_n$ such that $-I_2^{\otimes n}\notin G$, and $M_0\in \mb M_{2^n}$, define
\begin{equation}
\mc V_{M_0}:= \text{span}\{gM_0g: g\in G\}.
\end{equation}
Then the span of $\mc V_{M_0}$ such that $\mc V_{M_0}$ is an operator system coincides with all the correctable errors outside the normalizer of $G$ plus identity. In other words,
\begin{equation}
\text{span}\{ \mc V_{M_0}: M_0\ \text{is\ such\ that}\ \mc V_{M_0}\ \text{is\ an\ operator\ system}\} = \text{span}\{(\mc P_n \backslash N(G)) \cup I_2^{\otimes n}\}.
\end{equation}
\end{theorem}

The paper is structured as follows: In Section 2 we cover some preliminary material and prove Theorem 1. In Section 3 we review the stabilizer formalism and prove Theorem 2. 

\section*{Acknowledgements}
The authors would like to express their extreme gratitude to the Illinois Geometry Lab at the University of Illinois at Urbana-Champaign, from which this work originated. R. Araiza and P. Wu would like to thank Thomas Sinclair for comments on an earlier draft of the manuscript. R. Araiza was funded as a JL Doob Research Assistant Professor during the writing of this manuscript.


\section{Noncommutative graph and quantum error-correcting codes}
Suppose $\mc H$ is a Hilbert space and $\mb B(\mc H)$ is the set of bounded operators; \cite{duan2012} first introduced the definition of a (quantum) noncommutative graph:
\begin{definition}
A noncommutative graph $\mc V$ is a linear subspace of $\mb B(\mc H)$, such that 
\begin{itemize}
\item $v\in \mc V$ implies $v^* \in \mc V$.
\item $I \in \mc V$.
\end{itemize}
\end{definition}
As mentioned in the introduction, noncommutative graphs are a special class of more general objects known as operator systems, whose theory as been greatly developed over the last few decades (see \cite{Paulsen02}). A \textit{quantum code} is defined as a subspace $\mc C \subseteq \mc H$. We have the following definition of a \textit{quantum error-correcting code} a la Knill-Laflamme \cite{KL97}:
\begin{definition}
We say that $\mc C \subseteq \mc H$ is a quantum error-correcting code for the noncommutative graph $\mc V$ if 
\begin{equation}\label{}
\dim P_{\mc C}\mc V P_{\mc C} = 1,
\end{equation}
where $P_{\mc C}: \mc H \to \mc C$ is the projection onto $\mc C$.
\end{definition}
Lending our hand in answering the main question stated in the introduction, we consider noncommutative graphs built from two ingredients:
\begin{itemize}
\item A group $G$ with a unitary representation $\pi: G \to \mb B(\mc H)$.
\item An operator $M_0 \in \mb B(\mc H)$.
\end{itemize}
A subspace $\mc V \subseteq \mb B(\mc H)$ given by $(G,\pi)$ and $M_0$ may be defined as 
\begin{equation}\label{nc graph_1}
\mc V:= \overline{span}\{\pi(g)M_0\pi(g)^*:g\in G\}.
\end{equation}
We point out that in general, some restrictions on $(G,\pi)$ and $M_0$ are required for the subspace in \eqref{nc graph_1} to be a concrete operator system, see Remark \ref{restriction:example} for an example.

In this note, we study compact (finite) $G$ with a projective unitary representation 
$$\pi:G \to \mb B(\mc H),$$
 $\text{dim}\ \mc H = d<\infty$, and a prefixed $M_0 \in \mb B(\mc H)$ such that 
\begin{assumption}
$\mc V$ defined by \eqref{nc graph_1} is an operator system. 
\end{assumption}
Given a group $G$, we will always assume we have a fixed representation $\pi: G \to \mathbb (H)$, and $M_0$ is chosen such that \eqref{nc graph_1} is an operator system. At times we will denote $\mc V$ as $\mc V_{M_0}$ to emphasize the dependence on the operator $M_0$. 
\begin{remark}\label{restriction:example}
\noindent Note that in \cite{amosov2018_1, amosov2018_2}, the authors start with a positive operator $M_0$ and 
\begin{equation}\label{identity condition}
\int_{G} \pi(g)M_0\pi(g)^*d\mu(g) = I,
\end{equation}
where $\mu$ is the Haar measure on $G$. On the other hand, our assumption is weaker. For example, if $G = \{I_2,X\} \subseteq \mb M_2$ is the group generated by Pauli-X operator where $X = \begin{pmatrix}
0& 1 \\
1 & 0
\end{pmatrix}$, and the representation is given by $\pi: G \to U(2), g\mapsto g$. Then one can show by direct calculation that 
\begin{equation}
\begin{aligned}
& \mc V_{M_0}\ \text{is\ an\ operator\ system} \\ 
\iff &\ M_0 = c_0 I + c_1 Y + c_2 Z, c_0\neq 0, c_1,c_2 \in \mb C,\ \text{such\ that}\ \exists c\in \mb C, c_i = c \overline{c_i}, i=1,2.
\end{aligned}
\end{equation}
Thus our class of $\mc V_{M_0}$ is larger because $M_0$ can even be non-self-adjoint.
\end{remark}
In \cite[Proposition 1]{amosov2018_2}, for an Abelian group $G$, assuming all the unitaries have a common eigenspace, it was shown that the projection onto that eigenspace is an anticlique. We can generalize that result in our setting. 

Suppose $G$ is Abelian and $\pi: G \to \mb B(\mc H)$ is a finite dimensional projective unitary representation with $\text{dim}\ \mc H = d$, then for any $g_1,g_2 \in G$, we have
$$\pi(g_1)\pi(g_2) = \pi(g_2)\pi(g_1).$$
Therefore, $\{\pi(g):g\in G\}$ can be diagonalized simultaneously, i.e., there exists a basis $\{\ket{e_j}:1\le j \le d\}$, such that 
\begin{equation}
\pi(g) = \sum_{j=1}^d \lambda_j(g) \ket{e_j}\bra{e_j}
\end{equation}
For each $g\in G$, suppose $J_g$ is the index set such that for any $r,s\in J_g, r\neq s$, we have $\lambda_r(g)\neq \lambda_s(g)$, i.e., it is the index set of different eigenvalues. Then the spectral decomposition can be given as 
\begin{equation}
\pi(g) = \sum_{i \in J_g}\lambda_i(g)P_i(g),
\end{equation}
where $$P_i(g) = \sum_{j: \lambda_j(g) = \lambda_i(g)}\ket{e_j}\bra{e_j}$$
are disjoint projections onto the eigenspace corresponding to the eigenvalue $\lambda_i(g)$.
\begin{theorem}\label{General}
Suppose $\mc V_{M_0}$ defined by \eqref{nc graph_1} is an operator system and $G$ is Abelian. Then for any sequence $\{i_g \in J_g\}_{g\in G}$, if $$P = \prod_{g\in G}P_{i_g}$$
has rank no less than 2, then $P$ is an anticlique for $\mc V$.
\end{theorem}
\begin{proof}
We need to prove that for any $g \in G$, we have $$P\pi(g)M_0\pi(g)^*P = c(g)P$$
for some constant $c(g)$ only depending on $g$. In fact, 
\begin{align*}
& P \pi(g)M_0\pi(g)^* P = P\sum_{i\in J_g}\lambda_i(g)P_i(g) M_0\sum_{j\in J_g}\overline{\lambda_j(g)}P_i(g) P 
\end{align*}
Note that given $g \in G$, for $i,j\in J_g$, we have $P_i(g)P_j(g) = \delta_{i,j}P_i(g)$. Therefore, for any $i\in J_g$,
\begin{align*}
PP_i(g) = P_i(g)P = 
\begin{cases}
P,\ \text{if}\ i = i_g;\\
0,\ \text{otherwise}.
\end{cases}
\end{align*}
Then
\begin{align}\label{Step I}
& P\sum_{i\in J_g}\lambda_i(g)P_i(g) M_0\sum_{j\in J_g}\overline{\lambda_j(g)}P_i(g) P  =  |\lambda_{i_g}(g)|^2  PM_0P.
\end{align}
Since $I \in \mc V_{M_0}$ and we are in a finite-dimensional setting, there exist $g_1,\cdots g_m \in G$ and $c_1,\cdots, c_m \in \mb C$ such that
\begin{align*}
\sum_{r=1}^m c_r \pi(g_r)M_0\pi(g_r)^* = I.
\end{align*}
Multiplying $P$ from left and right and using the fact that 
\begin{align*}
P \pi(g)M_0\pi(g)^* P = |\lambda_{i_g}(g)|^2  PM_0P,\ \forall g\in G,
\end{align*}
we have
\begin{align}\label{Step II}
P = \sum_{r=1}^m c_r P\pi(g_r)M_0\pi(g_r)^*P = \sum_{r=1}^m c_r |\lambda_{i_{g_r}}(g_r)|^2  PM_0P. 
\end{align}
Plug \eqref{Step II} into \eqref{Step I}, we have
\begin{align*}
P \pi(g)M_0\pi(g)^* P & = P\sum_{i\in J_g}\lambda_i(g)P_i(g) M_0\sum_{j\in J_g}\overline{\lambda_j(g)}P_i(g) P \\
& = |\lambda_{i_g}(g)|^2  PM_0P \\
& = \frac{|\lambda_{i_g}(g)|^2}{\sum_{r=1}^m c_r |\lambda_{i_{g_r}}(g_r)|^2}P =: c(g)P
\end{align*}
which shows that $P$ is a valid anticlique.
\end{proof}

\section{Stabilizer formalism and noncommutative graphs}
The \emph{stabilizer formalism} first presented in \cite{gottesman1997} involves an Abelian subgroup $G$ of $\mc P_n$, which is the Pauli group on $n$ qubits, such that $-I_2^{\otimes n}\notin G$. Denote 
\begin{align}\label{Notation of Pauli}
\sigma_0 = I_2,\sigma_1 = Z =\begin{pmatrix}
1 & 0 \\
0 & -1
\end{pmatrix}, \sigma_2 = X= \begin{pmatrix}
0 & 1 \\
1 & 0
\end{pmatrix}, \sigma_3 = Y = \begin{pmatrix}
0 & i \\
-i & 0
\end{pmatrix}.
\end{align}
For simplicity of notation, and as is common in the field, we let $X_j$ denote
\begin{align*}
X_j = I_2\otimes \cdots I_2 \otimes \underbrace{X}_{j-th\ component} \otimes I_2 \otimes \cdots \otimes I_2
\end{align*}
and similarly for $Y_j,Z_j$ for $1\leq j\leq n$. The Pauli group $\mc P_n$ is defined by 
\begin{equation}
\begin{aligned}
\mc P_n& = \langle X_j,Y_j,Z_j: 1\le j \le n \rangle \\
& = \{c \sigma_{j_1} \otimes \sigma_{j_2} \otimes \cdots \otimes \sigma_{j_n}: c = \pm 1,\pm i, 0\le j_1,\cdots, j_n \le 3\}
\end{aligned}
\end{equation}
The stabilizer code is defined as follows:
\begin{definition}
For any Abelian subgroup $G \subseteq \mc P_n$ such that $-I_2^{\otimes n}\notin G$, the stabilizer code $\mc C_G$ is defined by
\begin{equation}
\mc C_G:= \text{span}\{\ket{\psi}: g\ket{\psi} = \ket{\psi} \forall g\in G\}.
\end{equation}
\end{definition}
The well-known theorem of stabilizer formalism is the following \cite{gottesman1997, nielsen}:
\begin{theorem}\label{Stabilizer: classical}
For any Abelian subgroup $G \subseteq \mc P_n$ such that $-I_2^{\otimes n}\notin G$, let $E\in \mb M_{2^n}$ is an operator and denote by $P$, the projection onto the stabilizer code $\mc C_G$. Then
\begin{equation}
PEP = c(E)P \iff E \in \text{span}\{(\mc P_n \backslash N(G)) \cup G\},
\end{equation}
where $N(G) = \{h\in \mc P_n: hGh^{-1} = G\}$.
\end{theorem}
In the framework of noncommutative graphs, the normalizer $N(G)$ also plays a great role, and we can essentially recover the stabilizer formalism via the following theorem:
\begin{theorem}\label{Stabilizer: new}
For any Abelian subgroup $G \subseteq \mc P_n$ such that $-I_2^{\otimes n}\notin G$, and $M_0\in \mb M_{2^n}$, define
\begin{equation}
\mc V_{M_0}:= \text{span}\{gM_0g: g\in G\}.
\end{equation}
Then the span of $\mc V_{M_0}$ such that $\mc V_{M_0}$ is an operator system coincides with all the correctable errors outside the normalizer of $G$ plus identity. In other words,
\begin{equation}
\text{span}\{ \mc V_{M_0}: M_0\ \text{is\ such\ that}\ \mc V_{M_0}\ \text{is\ an\ operator\ system}\} = \text{span}\{(\mc P_n \backslash N(G)) \cup I_2^{\otimes n}\}.
\end{equation}
\end{theorem}
\begin{remark}
Our noncommutative graph can recover all the detectable errors $E$ which are not commuting with $G$. The errors commuting with $G$ cannot be reflected in $\mc V_{M_0}$.
\end{remark}
Fix an Abelian subgroup $G \subseteq \mc P_n$ such that $-I_2^{\otimes n}\notin G$, we choose the representation $\pi:G \to \mb B(\mc H), g\mapsto g$. $M_0\in \mb B(\mc H)$ is prefixed. Then the linear subspace $\mc V_{M_0}$ is given by
\begin{equation}
\mc V_{M_0}:= \text{span}\{\pi(g)M_0\pi(g)^*: g\in G\} = \text{span}\{gM_0g: g\in G\},
\end{equation}
where the last equality follows from the fact that if $-I_2^{\otimes n}\notin G$, then $\forall g\in G, g^* = g$. We will adopt the following the index sets given by  
\begin{equation}\label{index set}
\mc I_0:= \{0,1\}, \mc I_1 := \{2,3\}.
\end{equation} 
then the following characterization is realized:
\begin{lemma}\label{Characterization of V_M}
Suppose $G = \langle Z_1,\cdots, Z_s \rangle$ for some $1\le s \le n$. Then $\mc V_{M_0}$ is an operator system if and only if 
\begin{equation}\label{statement: M_0}
\begin{aligned}
M_0 & = \alpha_{00\cdots 0}I_2^{\otimes n} + \sum_{\exists 1\le r\le s, j_r \in \mc I_1} \sum_{j_{s+1},\cdots, j_n=0}^3 \alpha_{j_1\cdots j_n}\sigma_{j_1}\otimes\cdots \otimes \sigma_{j_n} \\
& =  \alpha_{00\cdots 0}I_2^{\otimes n} + \sum_{\substack{i_1,\cdots, i_s = 0 \\ i_1+\cdots + i_s\neq 0}}^1  \sum_{\substack{j_r\in \mc I_{i_r}\\1\le r\le s}}\sum_{j_{s+1},\cdots, j_n=0}^3 \alpha_{j_1\cdots j_n}\sigma_{j_1}\otimes\cdots \otimes \sigma_{j_n},
\end{aligned}
\end{equation}
for $\alpha_{00\cdots 0} \neq 0$, and 
\begin{equation}\label{statement: M_0_1}
\begin{aligned}
&\forall i_1,\cdots, i_s=0,1,\ i_1+\cdots + i_s\neq 0, \exists c_{i_1,\cdots, i_s}\in \mb C, s.t., \\
&\overline{\alpha_{j_1\cdots j_n}} = c_{i_1,\cdots, i_s} \alpha_{j_1\cdots j_n}, \forall j_r\in I_{i_r}: 1\le r\le s;\ j_{s+1},\cdots, j_n=0,1,2,3.
 \end{aligned}
\end{equation}
\end{lemma}
\begin{proof}
Assume  $M_0$ is given by
\begin{equation}
M_0 = \ \sum_{j_{1},\cdots, j_n=0}^3 \alpha_{j_1\cdots j_n}\sigma_{j_1}\otimes\cdots \otimes \sigma_{j_n},\ \alpha_{j_1\cdots j_n}\in \mb C.
\end{equation}
\textbf{Necessity: suppose $\mc V_{M_0}$ is an operator system.} 

\noindent \textbf{Step I: Implication of $I_2^{\otimes n} \in \mc V_{M_0}$}: Firstly, $I_2^{\otimes n} \in \mc V_{M_0}$ implies that there exists $c_{i_1\cdots i_s}\in \mb C, i_1,\cdots, i_s=0,1$, such that 
\begin{equation}
\begin{aligned}
I_2^{\otimes n} & = \sum_{i_1,\cdots, i_s=0}^1 c_{i_1\cdots i_s}(Z_1^{i_1}\cdots Z_s^{i_s})M_0(Z_1^{i_1}\cdots Z_s^{i_s}) \\
& = \sum_{i_1,\cdots, i_s=0}^1 c_{i_1\cdots i_s}\sum_{j_1,\cdots, j_s=0}^3\sum_{j_{s+1},\cdots, j_n=0}^3 \alpha_{j_1\cdots j_n}Z^{i_1}\sigma_{j_1}Z^{i_1}\otimes \cdots \otimes Z^{i_s}\sigma_{j_s}Z^{i_s} \otimes \sigma_{j_{s+1}}\otimes \cdots \otimes \sigma_{j_n} \\
& = \sum_{i_1,\cdots, i_s=0}^1 c_{i_1\cdots i_s}\sum_{j_1,\cdots, j_s \in \mc I_0}\sum_{j_{s+1},\cdots, j_n=0}^3 \alpha_{j_1\cdots j_n}\sigma_{j_1}\otimes \cdots \otimes \sigma_{j_n}\\
& + \sum_{i_1,\cdots, i_s=0}^1 c_{i_1\cdots i_s}\sum_{\exists 1\le r\le s, j_r \in \mc I_1}\sum_{j_{s+1},\cdots, j_n=0}^3 \alpha_{j_1\cdots j_n}Z^{i_1}\sigma_{j_1}Z^{i_1}\otimes \cdots \otimes Z^{i_s}\sigma_{j_s}Z^{i_s} \otimes \sigma_{j_{s+1}}\otimes \cdots \otimes \sigma_{j_n}.
\end{aligned}
\end{equation}
For the last equality, we used the fact that for $i=0,1,\ j=\mc I_0\cup \mc I_1$ we have
\begin{equation}
 Z^i \sigma_j Z^{i} = \begin{cases}
 \sigma_j, &j\in\mc I_0; \\
 (-1)^i\sigma_j, &j \in \mc I_1.
\end{cases}
\end{equation}
Moreover, for any given $i_1,\cdots, i_s=0,1$, we have 
\begin{align*}
& \sum_{\exists 1\le r\le s, j_r \in \mc I_1}\sum_{j_{s+1},\cdots, j_n=0}^3 \alpha_{j_1\cdots j_n}Z^{i_1}\sigma_{j_1}Z^{i_1}\otimes \cdots \otimes Z^{i_s}\sigma_{j_s}Z^{i_s} \otimes \sigma_{j_{s+1}}\otimes \cdots \otimes \sigma_{j_n} \\
& = \sum_{k=1}^s \sum_{1\le r_1<\cdots < r_k\le s} \sum_{\substack{j_{r_1},\cdots, j_{r_k} \in \mc I_1\\ j_r\in \mc I_0, r\neq r_1,\cdots, r_k}}\sum_{j_{s+1},\cdots, j_n=0}^3 (-1)^{i_{r_1}+ \cdots i_{r_k}}\alpha_{j_1\cdots j_n} \sigma_{j_1}\otimes \cdots \otimes \sigma_{j_n}.
\end{align*}
Since $\{\sigma_{j_1}\otimes \cdots \otimes \sigma_{j_n}: 0\le j_1,\cdots, j_n \le 3\}$ form an orthonormal basis, we have
\begin{equation}\label{restriction}
    a\begin{aligned}
& \sum_{i_1,\cdots, i_s=0}^1 c_{i_1\cdots i_s} \alpha_{00\cdots 0} = 1, \\
& \sum_{i_1,\cdots, i_s=0}^1 c_{i_1\cdots i_s} \alpha_{j_1\cdots j_n} = 0,\ \forall j_1,\cdots,j_s\in \mc I_0, j_1+\cdots+j_s \neq 0,\\
& \sum_{i_1,\cdots, i_s=0}^1 c_{i_1\cdots i_s} (-1)^{i_{r_1}+ \cdots i_{r_k}} \alpha_{j_1\cdots j_s j_{s+1}\cdots j_{n}} =0, \\
& \ \forall j_{r_1},\cdots, j_{r_k} \in \mc I_1, j_r\in \mc I_0, r\neq r_1,\cdots, r_k\ \text{for\ some}\ 1\le r_1<\cdots < r_k\le s.
\end{aligned}
\end{equation}
For the existence of $c_{i_1\cdots i_s}\in \mb C, i_1,\cdots, i_s=0,1$, note that $\sum_{i_1,\cdots, i_s=0}^1 c_{i_1\cdots i_s} \alpha_{00\cdots 0} = 1$ implies $\sum_{i_1,\cdots, i_s=0}^1 c_{i_1\cdots i_s} \neq 0$ thus 
\begin{equation}\label{restriction:1}
\alpha_{j_1\cdots j_n} = 0,\ j_1,\cdots, j_s\in \mc I, j_1+\cdots+j_s \neq 0.
\end{equation}
Moreover, note that
\begin{align*}
\begin{cases}
\sum_{i_1,\cdots, i_s=0}^1 c_{i_1\cdots i_s} = \frac{1}{\alpha_{00\cdots 0}} \neq 0, \\
\sum_{i_1,\cdots, i_s=0}^1 (-1)^{i_{r_1}+ \cdots i_{r_k}} c_{i_1\cdots i_s} = 0,\forall 1\le r_1<\cdots < r_k\le s.
\end{cases}
\end{align*}
has a unique solution. Thus there is no requirement for $\alpha_{j_1,\cdots, j_n}$ if at least one $j_1,\cdots, j_s$ is in $\mc I_1$.

In summary, by \eqref{restriction:1}, if $I_2^{\otimes n} \in \mc V_{M_0}$, $M_0$ must have the form
\begin{equation}\label{form of M_0}
M_0 = \alpha_{00\cdots 0}I_2^{\otimes n} + \sum_{\exists 1\le r\le s, j_r \in \mc I_1} \sum_{j_{s+1},\cdots, j_n=0}^3 \alpha_{j_1\cdots j_n}\sigma_{j_1}\otimes\cdots \otimes \sigma_{j_n},
\end{equation}
where $\alpha_{00\cdots 0} \neq 0$ and $\alpha_{j_1\cdots j_n}$ can be arbitrary complex numbers if $\exists 1\le r\le s, j_r \in \mc I^c$. 

\noindent \textbf{Characterization of $\mc V_{M_0}$ with the form $\mc V_{M_0} = \text{span}\{I,A_i:1\le i \le m\}:$} 

\noindent Note that if $M_0$ is given by \eqref{form of M_0}, for any $x\in \mc V_{M_0}$, there exist $c_{i_1\cdots i_s} \in \mb C$: 
\begin{equation}\label{rep:M_0}
\begin{aligned}
x & = \sum_{i_1,\cdots, i_s=0}^1 c_{i_1\cdots i_s}(Z_1^{i_1}\cdots Z_s^{i_s})M_0(Z_1^{i_1}\cdots Z_s^{i_s}) \\
& = \sum_{i_1,\cdots, i_s=0}^1 c_{i_1\cdots i_s} \alpha_{00\cdots 0}I_2^{\otimes n} \\
& +  \sum_{i_1,\cdots, i_s=0}^1 c_{i_1\cdots i_s} \sum_{\exists 1\le r\le s, j_r \in \mc I_1} \sum_{j_{s+1},\cdots, j_n=0}^3 \alpha_{j_1\cdots j_n}Z^{i_1}\sigma_{j_1}Z^{i_1} \otimes\cdots \otimes Z^{i_s}\sigma_{j_s}Z^{i_s} \otimes \sigma_{j_{s+1}}\otimes \cdots \otimes \sigma_{j_n}.
\end{aligned}
\end{equation}
For any fixed $i_1,\cdots, i_s = 0,1$, we have
\begin{equation}\label{rewrite:M_0}
\begin{aligned}
& \sum_{\exists 1\le r\le s, j_r \in \mc I_1} \sum_{j_{s+1},\cdots, j_n=0}^3 \alpha_{j_1\cdots j_n}Z^{i_1}\sigma_{j_1}Z^{i_1} \otimes\cdots \otimes Z^{i_s}\sigma_{j_s}Z^{i_s} \otimes \sigma_{j_{s+1}}\otimes \cdots \otimes \sigma_{j_n} \\
& = \sum_{k=1}^s \sum_{1\le r_1<\cdots<r_k\le s} \sum_{\substack{j_{r_1},\cdots, j_{r_k}\in \mc I_1 \\ j_r \in \mc I_0,r\neq r_1,\cdots,r_k}}\sum_{j_{s+1},\cdots, j_n=0}^3 \alpha_{j_1\cdots j_n}Z^{i_1}\sigma_{j_1}Z^{i_1} \otimes\cdots \otimes Z^{i_s}\sigma_{j_s}Z^{i_s} \otimes \sigma_{j_{s+1}}\otimes \cdots \otimes \sigma_{j_n} \\
& = \sum_{k=1}^s \sum_{1\le r_1<\cdots<r_k\le s} (-1)^{i_{r_1}+\cdots i_{r_k}}\sum_{\substack{j_{r_1},\cdots, j_{r_k}\in \mc I_1 \\ j_r \in \mc I_0,r\neq r_1,\cdots,r_k}}\sum_{j_{s+1},\cdots, j_n=0}^3 \alpha_{j_1\cdots j_n} \sigma_{j_1}\otimes \cdots \otimes \sigma_{j_n}.
\end{aligned}
\end{equation} 
If we denote 
\begin{equation}\label{bijective}
\wt c_{r_1 \cdots r_k} = \sum_{i_1,\cdots, i_s=0}^1 (-1)^{i_{r_1}+\cdots i_{r_k}}c_{i_1\cdots i_s},\ \wt c_{\underbrace{00\cdots 0}_{s}} = \sum_{i_1,\cdots, i_s=0}^1 c_{i_1\cdots i_s}
\end{equation}
and note that there is a one-to-one correspondence between the index sets with $2^s - 1$ elements:
\begin{align*}
& \{(r_1,\cdots,r_k):1\le k \le s, 1\le r_1<\cdots<r_k \le s\}\\
& \text{and}\ \{(u_1,\cdots,u_s):u_1,\cdots, u_s=0,1; u_1+\cdots +u_s\neq 0\},
\end{align*}
then the sum in \eqref{rep:M_0}, via \eqref{rewrite:M_0} and \eqref{bijective}, can be rewritten as
\begin{align*}
x & = \sum_{i_1,\cdots, i_s=0}^1 c_{i_1\cdots i_s} \alpha_{00\cdots 0}I_2^{\otimes n} \\
& +  \sum_{i_1,\cdots, i_s=0}^1 c_{i_1\cdots i_s} \sum_{\exists 1\le r\le s, j_r \in \mc I_1} \sum_{j_{s+1},\cdots, j_n=0}^3 \alpha_{j_1\cdots j_n}Z^{i_1}\sigma_{j_1}Z^{i_1} \otimes\cdots \otimes Z^{i_s}\sigma_{j_s}Z^{i_s} \otimes \sigma_{j_{s+1}}\otimes \cdots \otimes \sigma_{j_n} \\
& = \wt c_{00\cdots 0} \alpha_{00\cdots 0}I_2^{\otimes n} + \sum_{\substack{u_1,\cdots, u_s=0 \\ u_1+\cdots +u_s\neq 0}}^1 \wt c_{u_1,\cdots, u_s} \sum_{\substack{j_r \in \mc I_{u_r}\\ 1\le r \le s}}\sum_{j_{s+1},\cdots, j_n=0}^3  \alpha_{j_1\cdots j_n}\sigma_{j_1}\otimes \cdots \otimes \sigma_{j_n}.
\end{align*}
Since the choice of $c_{i_1\cdots i_s},\ i_1,\cdots, i_s=0,1$ is arbitrary, and by the definition of \eqref{bijective} and the one-to-one correspondence between $\{(r_1,\cdots,r_k):1\le k\le s,1\le r_1<\cdots < r_k\le s\}$ and $\{(u_1,\cdots,u_s): u_1,\cdots, u_s=0,1; u_1+\cdots +u_s\neq 0\}$, 
\begin{equation}
\wt c_{u_1,\cdots, u_s},\ u_1,\cdots,u_s=0,1
\end{equation}
can be arbitrary complex numbers, thus
\begin{align}\label{V_M}
\mc V_{M_0} & = \text{span}\{I_2^{\otimes n}, \sum_{\substack{j_r \in \mc I_{u_r}\\ 1\le r \le s}}\sum_{j_{s+1},\cdots, j_n=0}^3  \alpha_{j_1\cdots j_n}\sigma_{j_1}\otimes \cdots \otimes \sigma_{j_n}: u_1,\cdots, u_s=0,1; u_1+\cdots +u_s\neq 0\}.
\end{align}
\textbf{Step II: Implication of $\mc V_{M_0}$ to be $*$-closed:} From the characterization of $\mc V_{M_0}$, we know that for any $M_0$ given by \eqref{form of M_0}, 
\begin{equation}
\mc V_{M_0} = \text{span}\{I,A_1,\cdots, A_m\}
\end{equation}
for some $m\le 2^s$, and for each $1\le i\le m$, $A_i = \sum_{j\in J_i}a_{j}^i e_j^i$ where $a_j^i\in\mb C$ and $\{e_j^i:1\le i\le m, j\in J_i\}$ form an orthonormal set of $\mb B(\mc H)$ with $\mc H = (\mb C^2)^{\otimes n}$. Then it is straightforward to check that $\mc V_{M_0}$ is $*$-closed if and only if 
\begin{equation}\label{star closed}
\forall 1\le i \le m,\ \exists c_i \in \mb C,\ s.t.,\ \overline{a_j^i}=c_i a_j^i,\ \forall j\in J_i.
\end{equation}
Translating \eqref{star closed} into our setting, we get 
\begin{equation}\label{star closed: translated}
\begin{aligned}
&\forall i_1,\cdots, i_s=0,1,\ i_1+\cdots + i_s\neq 0, \exists c_{i_1,\cdots, i_s}\in \mb C, s.t. \\
&\overline{\alpha_{j_1\cdots j_n}} = c_{i_1,\cdots, i_s} \alpha_{j_1\cdots j_n}, \forall j_r\in I_{i_r}: 1\le r\le s;\ j_{s+1},\cdots, j_n=0,1,2,3.
\end{aligned}
\end{equation}

\textbf{Sufficiency: $M_0$ given by \eqref{statement: M_0} and \eqref{statement: M_0_1} implies $\mc V_{M_0}$ is an operator system.}

\noindent $I_2^{\otimes n}\in \mc V_{M_0}$ since $\alpha_{00\cdots 0}\neq 0$. Moreover, we note that \eqref{star closed} implies $\mc V_{M_0}$ is $*$-closed.
\end{proof}

\textbf{Proof of Theorem \ref{Stabilizer: new}:}
\begin{proof}
First, recall our convention that 
\begin{align*}
\sigma_0 = I_2,\ \sigma_1 = Z,\ \sigma_2=X,\ \sigma_3 = Y.
\end{align*}
If $G = \langle Z_1,\cdots, Z_s\rangle$, by the characterization \eqref{V_M} with restriction \eqref{star closed: translated} in Lemma \ref{Characterization of V_M}, we know that 
\begin{align*}
& \text{span}\{ \mc V_{M_0}: M_0\ \text{is\ such\ that}\ \mc V_{M_0}\ \text{is\ an\ operator\ system}\} \\
& = \text{span}\{I_2^{\otimes n}, \sigma_{j_1}\otimes \cdots \otimes \sigma_{j_n}:\ j_1,\cdots, j_n = 0,1,2,3,\ \exists 1\le r\le s,\ j_r = 2,3\}.
\end{align*}
Also note that if $I\notin G$, we have $N(G) = Z(G)$ where $Z(G) = \{g\in G: gh=hg,\forall h\in G\}$ is the centralizer. We have 
\begin{equation}
N(G) = \{c \sigma_{j_1}\otimes \cdots \otimes \sigma_{j_n}: j_1,\cdots, j_s = 0,1,\ j_{s+1},\cdots, j_n = 0,1,2,3,\ c=\pm 1,\pm i\}.
\end{equation}
Then we arrive at the conclusion that 
\begin{align*}
& \text{span}\{ \mc V_{M_0}: M_0\ \text{is\ such\ that}\ \mc V_{M_0}\ \text{is\ an\ operator\ system}\} \\
& = \text{span}\{I_2^{\otimes n}, \sigma_{j_1}\otimes \cdots \otimes \sigma_{j_n}:\ j_1,\cdots, j_n = 0,1,2,3,\ \exists 1\le r\le s,\ j_r = 2,3\} \\
& = \text{span}\{I_2^{\otimes n}, \mc P_n \backslash N(G)\}.
\end{align*}
If $G$ is any Abelian subgroup of $\mc P_n$ such that $-I_2^{\otimes n} \notin G$, then it is well-known(see e.g. \cite{gottesman1997}) that there exists a global unitary $U: \mc H\to \mc H$ such that \begin{equation}
G = \langle \wt Z_1,\cdots, \wt Z_s\rangle
\end{equation}
where $\wt Z_i = U Z_i U^*$. For \footnote{It is also known as ``logical" basis}{the new basis} of $\mc H$ given by
\begin{equation}
\{\ket{(i_1)_L\cdots (i_n)_L}: i_1,\cdots, i_n=0,1\},\ \ket{(i_1)_L\cdots (i_n)_L} = U \ket{i_1\cdots i_n},
\end{equation} 
$\wt Z_i$ acts as the Pauli $Z$ operator on the $i$-th ``logical" qubit. Similarly, we can define $\wt X_i = UX_iU^*,\wt Y_i = UY_iU^*, \wt Z_i = UZ_iU^*$ for $1\le i\le n$. Following the same argument as before, we have
\begin{align*}
& \text{span}\{ \mc V_{M_0}: M_0\ \text{is\ such\ that}\ \mc V_{M_0}\ \text{is\ an\ operator\ system}\} \\
& = \text{span}\{I_2^{\otimes n}, \wt \sigma_{j_1}\otimes \cdots \otimes \wt \sigma_{j_n}:\ j_1,\cdots, j_n = 0,1,2,3,\ \exists 1\le r\le s,\ j_r = 2,3\} \\
& = \text{span}\{I_2^{\otimes n}, \mc P_n \backslash N(G)\}.
\end{align*}
\end{proof}

\end{document}